\documentclass[letterpaper, 10pt, onecolumn]{ieeeconf}  

\IEEEoverridecommandlockouts                              

\usepackage{amsmath}
\usepackage{amssymb}
\usepackage{amsfonts}
\usepackage{graphicx}
\usepackage{color}
\usepackage{mathptmx} 
\usepackage{times}

\usepackage{datetime}

\title{\LARGE \bf  
Weyl variations and local sufficiency of
linear observers in the mean square optimal coherent quantum filtering problem$^*$}

\author{Igor G. Vladimirov$^{\dagger}$
\thanks{$^*$This work is supported by the Australian Research Council.}
\thanks{$^\dagger$UNSW Canberra, Australia.
{\tt igor.g.vladimirov@gmail.com}.}
}

\DeclareMathAlphabet{\bit}{OML}{cmm}{b}{it}

\newtheorem{lem}{Lemma}
\newtheorem{theorem}{Theorem}

\def\ad{\mathrm{ad}}           
\def\<{\leqslant}           
\def\>{\geqslant}           

\def\Re{\mathrm{Re}}   
\def\Im{\mathrm{Im}}   

\def\mR{\mathbb{R}}    
\def\mC{\mathbb{C}}    
\def\mH{\mathbb{H}}    
\def\cH{\mathcal{H}}    
\def\cQ{\mathcal{Q}}    
\def\Tr{\mathrm{Tr}}   
\def\rT{\mathrm{T}}    

\def\bE{\mathbf{E}}    

\def\bra{\langle}
\def\ket{\rangle}

\def\Bra{\left\langle }
\def\Ket{\right\rangle }

\def\re{\mathrm{e}}    

\def\rd{\mathrm{d}}    

\def\d{\partial}    

\def\x{\times}
\def\ox{\otimes}

\def\cZ{\mathcal{Z}}

\def\b1{\mathbf{1}}

\def\bD{\mathbf{D}}

\def\bH{\mathbf{H}}

\def\bU{\mathbf{U}}

\def\bS{\mathbf{S}}
\def\cF{\mathcal{F}}
\def\cW{\mathcal{W}}

\def\cD{\mathcal{D}}

\def\cC{\mathcal{C}}

\def\cG{\mathcal{G}}
\def\cI{\mathcal{I}}
\def\cP{\mathcal{P}}

\def\cB{\mathcal{B}}
\def\cE{\mathcal{E}}

\def\cA{\mathcal{A}}

\def\bG{\mathbf{G}}
\def\bL{\mathbf{L}}

\def\mS{\mathbb{S}}
\def\mA{\mathbb{A}}

\def\eps{\epsilon}
\def\Ups{\Upsilon}

\def\bOmega{\bit{\Omega}}
\def\bTheta{\bit{\Theta}}
\def\bJ{\mathbf{J}}
\def\bj{\mathbf{j}}
\def\sj{\mathsf{j}}
\def\sW{\mathsf{W}}
\def\sS{\mathsf{S}}
\def\fH{\mathfrak{H}}
\def\fF{\mathfrak{F}}
\def\cX{\mathcal{X}}
\begin{document}
\maketitle
\thispagestyle{empty}
\pagestyle{plain}

\centerline{In loving memory of my mother, Raisa Dmitrievna Vladimirova (10.05.1933 -- 30.05.2015)}\vskip1cm

\begin{abstract}
This paper is concerned with the coherent quantum filtering (CQF) problem, where a quantum observer is cascaded in a measurement-free fashion  with a linear quantum plant so as to minimize a mean square error of estimating the plant variables of interest. Both systems are governed by Markovian Hudson-Parthasarathy quantum stochastic differential equations driven by bosonic fields in vacuum state. These quantum dynamics are specified by the Hamiltonians and system-field coupling operators. We apply a recently proposed transverse Hamiltonian variational method to the development of  first-order necessary conditions of optimality for the CQF problem in a larger class of observers. The latter is obtained by perturbing the Hamiltonian and system-field coupling operators of a linear coherent quantum observer along linear combinations of unitary Weyl operators, whose role here resembles that  of the needle variations in the Pontryagin  minimum principle.
We show that if the observer is a stationary point of the performance functional in the class of linear observers, then it is also a stationary point with respect to the Weyl variations in the larger class of nonlinear observers.
 \end{abstract}

\section{INTRODUCTION}\label{sec:intro}

The dynamics of a wide class of open quantum systems, interacting with the
environment, can be described in the framework of the Hudson-Parthasarathy quantum stochastic calculus \cite{H_1991,HP_1984,P_1992}. This approach employs quantum stochastic differential equations (QSDEs) which are driven by a quantum mechanical analogue of the classical Wiener process \cite{KS_1991}. The quantum Wiener process models a heat reservoir of external fields and acts on a boson Fock space \cite{PS_1972,P_1992}. The drift vector and the dispersion matrix of the QSDEs depend on the system Hamiltonian and the system-field coupling operators. These energy operators specify the evolution of the system as a result of its internal dynamics influenced by the interaction with the environment.
The Hamiltonian and the coupling operators are usually modelled as functions
of the system variables. In particular, such functions can be polynomials or Weyl quantization integrals \cite{F_1989,SVP_2014}, which affects the  complexity of the resulting quantum system.

An important role in the linear quantum control and filtering theory \cite{DP_2010,JNP_2008,P_2010} is played by open quantum harmonic oscillators (OQHOs) \cite{EB_2005,GZ_2004} with quadratic Hamiltonians and linear system-field coupling operators. The linear-quadratic dependence of the energy operators on the system variables, in combination with the canonical commutation relations (CCRs) between the variables,  makes the OQHO dynamics linear (and  Gaussian in the case of vacuum fields and Gaussian initial states \cite{JK_1998,KRP_2010}).  Despite some similarities to the classical linear SDEs,   the coherent (that is, measurement-free) quantum counterparts \cite{NJP_2009, MJ_2012} to the classical  LQG control and filtering problems \cite{AM_1989,KS_1972} for OQHOs  are complicated by the physical realizability (PR) constraints. The latter  are associated with the state-space matrices of the QSDEs for fully quantum controllers or filters and are related, in particular, to the CCR preservation.

We mention one of the existing variational approaches \cite{VP_2013a,VP_2013b}  to the coherent quantum LQG (CQLQG) control and coherent quantum filtering (CQF) problems which develops optimality conditions using the Frechet differentiation of the LQG cost with respect to the state-space matrices. The quantum nature of the underlying problem  enters this approach  only through the PR constraints, with all the other aspects of the method being essentially ``classical''. The latter has certain advantages, such as practical applicability to the numerical optimization algorithms \cite{SVP_2015}. However, this approach is limited to linear controllers and filters, and the resulting optimality   conditions do not provide insights into whether nonlinear quantum controllers or filters can outperform the linear ones for linear quantum plants.

In the present paper, we consider a CQF problem, similar to \cite{MJ_2012,VP_2013b}, where a quantum observer is cascaded in a measurement-free fashion  with a linear quantum plant  so as to minimize a mean square error with which the observer variables approximate linear combinations of plant variables of interest. Both systems are governed by Markovian Hudson-Parthasarathy QSDEs driven by bosonic fields in vacuum state.
We employ a recently proposed fully quantum variational method of \cite{V_2015} based on using a transverse Hamiltonian, an auxiliary time-varying operator which encodes the propagation of perturbations through the unitary system-field evolution. We apply the transverse Hamiltonian approach to the development of  first-order necessary conditions of optimality for the CQF problem in a larger class of observers. The latter is obtained by perturbing the Hamiltonian and system-field coupling operators of a linear coherent quantum observer along linear combinations of the unitary Weyl operators \cite{F_1989}. Similar trigonometric polynomials of quantum variables have recently been used in \cite{SVP_2014} for modelling the uncertainty in system Hamiltonians. In the present paper, the Weyl variations play a different role which resembles that of the needle variations in the proof of the Pontryagin minimum principle \cite{PBGM_1962}. We show that if the observer is a stationary point of the cost functional in the class of linear observers, then it is also a stationary point with respect to the Weyl variations (with the latter leading to nonlinear observers). Therefore, in the mean square optimal CQF problem for linear quantum plants, linear coherent quantum observers are locally sufficient   at least in the sense of the Weyl variations of the energy operators.


The paper is organised as follows. Section~\ref{sec:not} outlines notation used in the paper. Section~\ref{sec:OQHO} specifies the class of open quantum stochastic  plants being considered. Section~\ref{sec:CQF} formulates the mean square optimal CQF problem. Section~\ref{sec:var} employs the transverse Hamiltonian approach in order to study the sensitivity of the performance criterion to perturbations in the energy operators of observers. Section~\ref{sec:linopt} applies these results to liner-quadratic perturbations of the observer energy operators and obtains the first-order necessary conditions of optimality in the class of linear observers.  Section~\ref{sec:Weyl} introduces the Weyl variations of the energy operators and establishes the main result of the paper that the stationarity of the quadratic cost functional with respect to linear-quadratic perturbations of a linear observer implies the  stationarity with respect to the Weyl variations. Section~\ref{sec:conc} provides concluding remarks.

\section{PRINCIPAL NOTATION}\label{sec:not}

In what  follows, $[A,B]:= AB-BA$ denotes the commutator of linear operators $A$ and $B$ on a common space. As a linear superoperator, the  commutator with a fixed operator $A$, is denoted by $\ad_A(\cdot):= [A,\cdot]$.
This extends to the commutator $(n\x m)$-matrix  $
    [X,Y^{\rT}]
    :=
    XY^{\rT} - (YX^{\rT})^{\rT} = ([X_j,Y_k])_{1\< j\< n,1\< k\< m}
$ for a vector $X$ of operators $X_1, \ldots, X_n$ and a vector $Y$ of operators $Y_1, \ldots, Y_m$.  Vectors are organized as columns unless indicated otherwise,  and the transpose $(\cdot)^{\rT}$ acts on matrices of operators as if their entries were scalars.
In application to such matrices, $(\cdot)^{\dagger}:= ((\cdot)^{\#})^{\rT}$ denotes the transpose of the entry-wise operator adjoint $(\cdot)^{\#}$. For complex matrices,  $(\cdot)^{\dagger}$ is the usual complex conjugate transpose  $(\cdot)^*:= (\overline{(\cdot)})^{\rT}$.
The subspaces of real symmetric, real antisymmetric and complex Hermitian matrices of order $n$ are denoted by $\mS_n$, $\mA_n$
 and
$
    \mH_n
    :=
    \mS_n + i \mA_n
$, respectively,  where $i:= \sqrt{-1}$ is the imaginary unit. The real and imaginary parts of a complex matrix are denoted by $\Re(\cdot)$ and $\Im(\cdot)$. These extend to matrices $M$ with operator-valued entries as $\Re M = \frac{1}{2}(M+M^{\#})$ and $\Im M = \frac{1}{2i}(M-M^{\#})$ which consist of self-adjoint operators. Also, $\bS(M):= \frac{1}{2}(M+M^{\rT})$ 
denotes the symmetrizer 
of square matrices. 
Positive (semi-) definiteness of matrices and the corresponding partial ordering  are denoted by ($\succcurlyeq$) $\succ$.  Also, $\mS_n^+$ and $\mH_n^+$ denote the sets of positive semi-definite real symmetric and complex Hermitian matrices of order $n$, respectively.
The tensor product of spaces or operators (in particular, the Kronecker product of matrices) is denoted by $\ox$. The tensor product $A\ox B$ of operators $A$ and $B$ acting on different spaces will sometimes be abbreviated as $AB$.
The identity matrix of order $n$ is denoted by $I_n$, while the identity operator on a space $H$ is denoted by $\cI_H$.
The Frobenius inner product of real or complex matrices  is denoted by
$
    \bra M,N\ket
    :=
    \Tr(M^*N)
$.
Also, $\|v\|_K:= \sqrt{v^{\rT}Kv}$ denotes the Euclidean (semi-)norm of a real vector $v$ associated with a real positive (semi-)definite symmetric matrix $K$.
The expectation $\bE \xi := \Tr(\rho \xi)$  of a quantum variable $\xi$ over a density operator $\rho$ extends entrywise to vectors and matrices of such variables. 
\section{QUANTUM PLANTS BEING CONSIDERED}\label{sec:OQHO}

We consider a quantum plant which is modelled as a quantum stochastic system interacting with $m$ external boson fields. The plant
 has $n$ dynamic variables $X_1(t), \ldots, X_n(t)$ which evolve in time $t\> 0$. The plant variables are self-adjoint operators on a composite plant-field Hilbert space $\cH\ox \cF$, where $\cH$ is the initial complex separable Hilbert space of the plant which provides a domain for $X_1(0), \ldots, X_n(0)$, and $\cF$ is a boson Fock space \cite{P_1992} for the action of quantum Wiener processes $W_1(t), \ldots, W_m(t)$. The latter are self-adjoint operators which model the external boson fields. The energetics of the plant-field interaction is specified by the plant Hamiltonian $H(t)$ and the plant-field coupling operators $L_1(t), \ldots, L_m(t)$ which are self-adjoint operators,  representable as time-invariant functions (for example, polynomials with constant coefficients or Weyl quantization integrals \cite{F_1989})
 of the plant variables $X_1(t), \ldots, X_n(t)$. Therefore, both $H(0)$ and $L_1(0), \ldots, L_m(0)$ act on the initial space $\cH$. Omitting the time arguments,  we assemble the plant and field variables, and the coupling operators into vectors:
\begin{equation}
\label{XWL}
    X:=
    {\begin{bmatrix}
        X_1\\
        \vdots\\
        X_n
    \end{bmatrix}},
    \qquad
    W:=
    {\begin{bmatrix}
        W_1\\
        \vdots\\
        W_m
    \end{bmatrix}},
    \qquad
    L:=
    {\begin{bmatrix}
        L_1\\
        \vdots\\
        L_m
    \end{bmatrix}}.
\end{equation}
For what follows, the plant is assumed to be an OQHO \cite{EB_2005}. More precisely,
the plant variables satisfy the Weyl CCRs
$
\re^{i(u+v)^{\rT} X} = \re^{i u^{\rT}\Theta v} \re^{iu^{\rT} X} \re^{iv^{\rT} X}
$ for all $u,v\in \mR^n$, where
$
    \re^{iu^{\rT} X}
$
is the unitary Weyl operator \cite{F_1989} on $\cH\ox \cF$ which is associated with (and inherits time dependence from) the plant variables. The Heisenberg  infinitesimal form of the Weyl CCRs is described by
\begin{equation}
\label{Theta}
    [X,X^{\rT}]
    =
    2i
    \Theta
\end{equation}
where the CCR matrix $\Theta \in \mA_n$ represents $\Theta \ox \cI_{\cH\ox \cF}$ and remains unchanged.
The Hamiltonian of the OQHO is a quadratic polynomial of the plant variables, and  the plant-field coupling operators in (\ref{XWL}) are linear functions of the variables:
\begin{equation}
\label{RN}
    H = \frac{1}{2}X^{\rT} R X,
    \qquad
    L
    =
    N X,
\end{equation}
where $R\in \mS_n$ is the energy matrix, and $N \in \mR^{m\x n}$ is the coupling matrix. The plant and the external fields form an isolated quantum system whose evolution  is described by a unitary operator $U(t)$ on $\cH\ox \cF$ driven by the fields and their interaction with the plant according to the QSDE \cite{HP_1984,P_1992}
\begin{align}
\nonumber
    \rd U
    & =
    -\Big(i(H_0\rd t + L_0^{\rT} \rd W) + \frac{1}{2}L_0^{\rT}\Omega L_0\rd t\Big)U\\
\label{dUHL}
    & =
    -U\Big(i(H\rd t + L^{\rT} \rd W) + \frac{1}{2}L^{\rT}\Omega L\rd t\Big),
\end{align}
with initial condition $U_0:=\cI_{\cH\ox \cF}$. The subscript $(\cdot)_0$ indicates the initial values of time-varying operators (or vectors and matrices thereof), so that $H_0:= H(0)$, $L_0:= L(0)$ and $U_0:= U(0)$, and the time arguments will often be omitted for brevity.
Due to the continuous tensor product structure of the Fock space \cite{PS_1972}, the future-pointing increments $\rd W$ commute with adapted processes (including $U$) taken at the same (or an earlier) moment of time.
 The matrix
$\Omega:= (\omega_{jk})_{1\<j,k\< m}\in \mH_m^+$ in (\ref{dUHL}) is the Ito matrix of the quantum Wiener process $W$:
\begin{equation}
\label{Omega}
    \rd W\rd W^{\rT} = \Omega \rd t,
    \qquad
    \Omega:= I_m + iJ,
\end{equation}
and the matrix $J\in \mA_m$ specifies the cross-commutations between the entries $W_1, \ldots, W_m$ of $W$:
\begin{equation}
\label{J}
    [\rd W, \rd W^{\rT}]
    =
    2iJ\rd t,
    \quad
    J := \bJ\ox I_{m/2},
    \quad
    \bJ
    :=
    {\begin{bmatrix}
        0& 1\\
        -1 & 0
    \end{bmatrix}},
\end{equation}
where the dimension $m$  is assumed to be even.
  The QSDE (\ref{dUHL}) corresponds to an important  particular case of open quantum dynamics when the scattering matrix is the identity matrix, and there is no photon exchange between the fields, thus eliminating the gauge processes \cite{P_1992} from consideration.
  A plant operator $\sigma_0$ on the initial space $\cH$ (which can be identified with its extension $\sigma\ox \cI_{\cF}$ to the plant-field space $\cH\ox \cF$) evolves to an operator $\sigma(t)$ on $\cH\ox \cF$ at time $t\>0$ according to the flow
\begin{equation}
\label{flow}
    \sigma(t)
    :=
    \sj_t(\sigma_0)
    =
    U(t)^{\dagger}(\sigma_0\ox \cI_{\cF})U(t).
\end{equation}
When it is applied to vectors and matrices of operators,  the flow  $\sj_t$ acts entrywise.  In view of the identity $(\sigma_0\ox \cI_{\cF})U = U\sigma$, which follows from (\ref{flow}) and the unitarity of $U$, the second equality in (\ref{dUHL})  employs the representation of the Hamiltonian and the coupling operators in terms of the flow:
\begin{equation*}
\label{HL}
    H(t) = \sj_t(H_0),
    \qquad
    L(t) = \sj_t(L_0).
\end{equation*}
Note that the flow $\sj_t$ depends on the energy operators $H_0$ and $L_0$ (or the energy and coupling matrices $R$ and $N$ in (\ref{RN}) in the case of OQHOs). More precisely, any perturbation of $H_0$ and $L_0$ (as functions of the fixed set of system variables $X_0$) modifies  the flow.
Now, the quantum adapted  process $\sigma$ in (\ref{flow}) satisfies the following Hudson-Parthasarathy QSDE \cite{HP_1984,P_1992}:
\begin{equation}
\label{dsigma}
    \rd \sigma = \cG(\sigma) \rd t - i[\sigma,L^{\rT}]\rd W,
    \quad
    \cG(\sigma) := i[H,\sigma] + \cD(\sigma).
\end{equation}
Here, $\cD$ is the Gorini-Kossakowski-Sudar\-shan-Lin\-d\-blad (GKSL) decoherence superoperator \cite{GKS_1976,L_1976} which acts on $\sigma$ as
\begin{align}
\nonumber
    \cD(\sigma)
    & :=
    \frac{1}{2}
    \big(
        L^{\rT}\Omega [\sigma,L]
        +
        [L^{\rT},\sigma]\Omega L
    \big)\\
\nonumber
    & =
    -[\sigma, L^{\rT}]\Omega L - \frac{1}{2} [L^{\rT}\Omega L, \sigma]\\
\label{cD}
    & =
    -[\sigma, X^{\rT}]N^{\rT}\Omega NX - \frac{1}{2} [X^{\rT}N^{\rT}\Omega NX, \sigma].
\end{align}
The last two equalities in (\ref{cD}) are convenient for the entrywise evaluation of $\cD$ at vectors of operators.
The superoperator $\cG$ in (\ref{dsigma}) is referred to as the GKSL generator. In application to
the vector $X$ of plant variables, the flow (\ref{flow}) acts entrywise as
\begin{equation*}
\label{X}
    X(t):=\sj_t(X_0) = U(t)^{\dagger}(X_0\ox \cI_{\cF})U(t).
\end{equation*}
 In view of (\ref{Theta}) and (\ref{RN}), the corresponding QSDE (\ref{dsigma}) takes the form
\begin{equation}
\label{dX}
    \rd X
    =
    \cG(X)\rd t -i[X, L^{\rT}]\rd W
    =
    AX\rd t + B\rd W,
\end{equation}
with the $n$-dimensional drift vector $\cG(X)= AX$ and the dispersion $(n\x m)$-matrix $-i[X, L^{\rT}]=-i[X, X^{\rT}]N^{\rT}= B$, where the matrices $A\in \mR^{n\x n}$ and $B\in \mR^{n\x m}$ are given by
\begin{equation}
\label{AB}
    A:= 2\Theta (R +N^{\rT}JN),
    \qquad
    B:= 2\Theta N^{\rT}.
\end{equation}
The interaction of the input field with the plant produces an $m$-dimensional  output field
\begin{equation}
\label{Y}
    Y(t):=
    {\begin{bmatrix}
        Y_1(t)\\
        \vdots\\
        Y_m(t)
    \end{bmatrix}}
    =
    U(t)^{\dagger}(\cI_{\cH}\ox W(t))U(t),
\end{equation}
where the plant-field unitary evolution is applied to the current input field variables (which is closely related to the innovation role of the quantum Wiener process $W$ in the QSDEs). The output field satisfies
the QSDE
\begin{equation}
\label{dY}
    \rd Y
    =
    2JL \rd t + \rd W
    =
    CX \rd t + \rd W,
\end{equation}
where $J$ is the matrix from (\ref{J}), $L$ is the vector of plant-field coupling operators in (\ref{XWL}) and (\ref{RN}), and the matrix $C\in \mR^{m\x n}$ is given by
\begin{equation}
\label{C}
    C := 2JN.
\end{equation}
\section{
COHERENT QUANTUM FILTERING PROBLEM}\label{sec:CQF}

Consider a measurement-free cascade connection of the quantum plant, described in Section~\ref{sec:OQHO}, with another open quantum system. The latter plays the role of a coherent quantum  observer and is driven by $\Pi Y$ (which is part  of the plant output $Y$ in (\ref{Y})) and a quantum Wiener process $\omega$ of even dimension $\mu$ on a boson Fock space $\fF$; see Fig.~\ref{fig:filtering}.
\begin{figure}[htbp]
\centering
\unitlength=0.8mm
\linethickness{0.2pt}
\begin{picture}(50.00,50.00)
    \put(20,20){\framebox(20,10)[cc]{{\small observer}}}
    \put(20,40){\framebox(20,10)[cc]{{\small plant}}}
    \put(20,25){\vector(-1,0){10}}
    \put(50,25){\vector(-1,0){10}}

    \put(30,40){\vector(0,-1){10}}
    \put(9,25){\makebox(0,0)[rc]{{\small$\eta$}}}
    \put(52,25){\makebox(0,0)[lc]{{\small$\omega$}}}
    \put(52,45){\makebox(0,0)[lc]{{\small$W$}}}
   \put(50,45){\vector(-1,0){10}}
    \put(32,35){\makebox(0,0)[lc]{{\small$\Pi Y$}}}
\end{picture}\vskip-15mm
\caption{The cascade connection of a quantum observer with a quantum plant, mediated by the field $\Pi Y$ and affected by the environment through the quantum Wiener processes $W$ and $\omega$. Also shown is the observer output $\eta$.
}
\label{fig:filtering}
\end{figure}
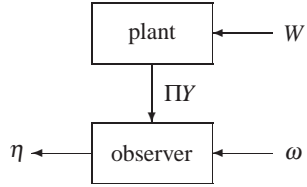
The matrix $\Pi\in \mR^{p\x m}$ is formed from conjugate pairs of rows of a permutation matrix of order $m$, with $p\< m$.
The observer has its own initial Hilbert space $\fH$, dynamic variables  $\xi_1, \ldots, \xi_{\nu}$ with a CCR matrix $\vartheta \in \mA_{\nu}$, and a  $(p+\mu)$-dimensional output field $\eta$:
\begin{equation}
\label{xiomlam}
    \xi :=
    {\begin{bmatrix}
        \xi_1 \\
        \vdots\\
        \xi_{\nu}
    \end{bmatrix}},
    \qquad
    \omega :=
    {\begin{bmatrix}
        \omega_1 \\
        \vdots\\
        \omega_{\mu}
    \end{bmatrix}},
    \qquad
    \eta :=
    {\begin{bmatrix}
        \eta_1 \\
        \vdots\\
        \eta_{p+\mu}
    \end{bmatrix}}.
\end{equation}
We denote the observer Hamiltonian by $\Gamma$, while the vectors of operators of coupling of the observer with the selected plant output $\Pi Y$ and the quantum Wiener process $\omega$ are denoted by
\begin{equation}
\label{PhiPsi}
    \Phi :=
    {\begin{bmatrix}
        \Phi_1 \\
        \vdots\\
        \Phi_p
    \end{bmatrix}},
    \qquad
    \Psi :=
    {\begin{bmatrix}
        \Psi_1 \\
        \vdots\\
        \Psi_{\mu}
    \end{bmatrix}},
\end{equation}
respectively.
The Hamiltonian $\Gamma$ and the coupling operators $\Phi_1, \ldots, \Phi_p$ and $\Psi_1, \ldots, \Psi_{\mu}$ are functions of the dynamic variables $\xi_1, \ldots, \xi_{\nu}$ of the observer and hence, commute with functions of the plant variables, including the plant Hamiltonian $H$ and the plant-field coupling operators in $L$.
The plant and the observer form a composite  open quantum stochastic system, whose vector $\cX$ of dynamic variables satisfies the CCRs
\begin{equation}
\label{cX}
  [\cX, \cX^{\rT}] = 2i\bTheta,
  \qquad
    \bTheta
    :=
    {\begin{bmatrix}
    \Theta & 0\\
    0 & \vartheta
    \end{bmatrix}},
    \qquad
    \cX
    :=
    {\begin{bmatrix}
        X\\
        \xi
    \end{bmatrix}}
\end{equation}
and is driven by a combined quantum Wiener process $\cW$ with the Ito table
\begin{equation}
\label{cW}
    \rd \cW\rd \cW^{\rT}
    =
    \bOmega \rd t,
    \qquad
    \bOmega:=
    {\begin{bmatrix}\Omega & 0\\
    0 & \mho
    \end{bmatrix}},
    \qquad
    \cW
    :=
    {\begin{bmatrix}
        W\\
        \omega
    \end{bmatrix}}.
\end{equation}
Here, $\mho$ is the Ito matrix  of the quantum Wiener process $\omega$  of the observer which is defined similarly to $\Omega$ in (\ref{Omega}) and (\ref{J}):
\begin{equation}
\label{Lambda}
    \rd \omega\rd \omega^{\rT} = \mho \rd t,
    \qquad
    \mho:= I_{\mu} + i\bJ\ox I_{\mu/2}.
\end{equation}
The Hamiltonian $\bH$ of the plant-observer system and the vector $\bL$ of operators of coupling with $\cW$ can be computed by using the quantum feedback  network formalism \cite{GJ_2009} as
\begin{align}
\label{bHbL}
    \bH = H + \Gamma + \Phi^{\rT} \Pi J L,
    \qquad
    \bL= {\begin{bmatrix}L + \Pi^{\rT}\Phi\\ \Psi\end{bmatrix}}.
\end{align}
%
%
%
%
%
%
While the plant dynamics, governed by (\ref{dX}) and (\ref{dY}),  remains unaffected by the observer, the dynamic variables of the latter in (\ref{xiomlam}) are governed by the QSDE
\begin{align}
\nonumber
    \rd \xi
    & =
    \bG(\xi)\rd t
    -
    i[\xi,\bL^{\rT}] \rd \cW\\
\label{dxi}
    & =
    (i[\Gamma,\xi] + \Delta(\xi))\rd t
    -
    i[\xi, \Phi^{\rT}] \Pi \rd Y
    -
    i[\xi, \Psi^{\rT}] \rd \omega.
\end{align}
Here,
\begin{equation}
\label{bG}
    \bG(\zeta):= i[\bH, \zeta] + \bD(\zeta)
\end{equation}
denotes the GKSL generator for the plant-observer system, and
\begin{equation}
\label{bD}
    \bD(\zeta)
    =
    -[\zeta, \bL^{\rT}]\bOmega \bL - \frac{1}{2} [\bL^{\rT}\bOmega \bL, \zeta]
\end{equation}
is  the corresponding decoherence superoperator, similar to (\ref{cD}).
In (\ref{dxi}), use is also made of the partial GKSL decoherence superoperator $\Delta$  which acts on the observer variables as
\begin{align}
\nonumber
    \Delta(\xi)
    = &
    -[\xi, \Phi^{\rT}]\Pi\Omega \Pi^{\rT}\Phi -[\xi, \Psi^{\rT}]\mho \Psi\\
\label{Delta}
    & - \frac{1}{2} [\Phi^{\rT}\Pi \Omega \Pi^{\rT} \Phi + \Psi^{\rT}\mho\Psi, \xi]
\end{align}
in view of (\ref{cX})--(\ref{bHbL}), with $\Pi\Omega \Pi^{\rT} \in \mH_p^+$ being the quantum Ito matrix of $\Pi W$.
Now, similarly to \cite{MJ_2012,VP_2013b}, we formulate a CQF problem  as the minimization of the steady-state mean square discrepancy
\begin{equation}
\label{CQF}
    \cZ
    :=
    \lim_{t\to +\infty} \bE Z(t) \longrightarrow \min,
    \qquad
    Z :=
    E^{\rT}E,
    \quad
    E:= F X - G\xi
\end{equation}
between $q$ linear combinations of the plant variables of interest and observer variables as specified by given matrices $F\in \mR^{q\x n}$ and $G \in \mR^{q\x \nu}$ (with $E$ being interpreted as the estimation error). Here, the quantum expectation $\bE(\cdot)$ is taken over the tensor product $\rho := \varpi \ox \upsilon$ of the initial quantum state $\varpi$ of the plant-observer system and the vacuum state $\upsilon$ in the composite boson Fock space $\cF \ox \fF$ for the external fields. Also, it is assumed that the plant-observer system is ergodic, whereby the limit in (\ref{CQF}) reduces to averaging over the invariant state of the system, provided the plant and observer variables satisfy an appropriate version of the uniform mean square integrability condition. 
The criterion process $Z$ in (\ref{CQF}) is a quadratic function of the vector $\cX$ from (\ref{cX}):
\begin{equation}
\label{ZCQF}
    Z :=
    \cX^{\rT}
    \cC^{\rT}
    \cC
    \cX,
    \qquad
    \cC
    :=
    {\begin{bmatrix}
        F & -G
    \end{bmatrix}}.
\end{equation}
The minimization in (\ref{CQF})
is carried out over the observer Hamiltonian $\Gamma$ and the vector $\Phi$ of the observer-plant coupling operators in (\ref{PhiPsi}), while $\Psi$ and all the dimensions are fixed.
This problem extends \cite{MJ_2012,VP_2013b} in that we do not restrict attention to linear observers even though the plant is an OQHO.

\section{INFINITESIMAL PERTURBATION ANALYSIS}\label{sec:var}

In order to develop first-order necessary conditions of optimality for the CQF problem (\ref{CQF}) in an extended class of observers, we will apply the transverse Hamiltonian variational method of \cite{V_2015} to the infinitesimal perturbation analysis of the performance criterion $\cZ$. To this end, suppose $\Gamma_0$ and $\Phi_0$ depend smoothly  (for example, linearly) on a small scalar parameter $\eps$ and are perturbed in the directions
\begin{equation}
\label{KM0CQF}
    K_0:= \Gamma_0',
    \qquad
    M_0:= \Phi_0',
\end{equation}
consisting of self-adjoint operators on the observer initial space $\fH$, representable as functions of the observer variables, with $(\cdot)':= \d_{\eps}(\cdot)|_{\eps = 0}$. The  corresponding perturbations of the plant-observer  Hamiltonian and the coupling operators in (\ref{bHbL}) are
\begin{equation}
\label{KMCQF}
    \bH_0' = K_0  + M_0^{\rT}\Pi JL_0,
    \qquad
    \bL_0'= {\begin{bmatrix}\Pi^{\rT}M_0\\ 0\end{bmatrix}}.
\end{equation}
The propagation of these  perturbations of the energy operators through the unitary evolution $\bU$ of the   plant-observer-field  system  on the space $\cH\ox \fH \ox \cF \ox \fF$ is encoded by the transverse Hamiltonian \cite{V_2015}. The latter is a time-varying self-adjoint operator defined by $Q:= i\bU^{\dagger}\bU'$, so that $\bU(t)' = -i\bU(t)Q(t)$ for all $t\>0$, with zero initial condition $Q_0 = 0$.  In view of (\ref{cW}), (\ref{KM0CQF}) and (\ref{KMCQF}),  the general QSDE obtained in \cite[Theorem 1]{V_2015} for the transverse Hamiltonian, takes the form:
\begin{align}
\nonumber
    \rd Q
    & =
        \Big(K  + M^{\rT}\Pi JL -\Im \Big(\bL^{\rT}\bOmega {\small\begin{bmatrix}\Pi^{\rT} M\\ 0\end{bmatrix}}\Big)\Big)\rd t
        +
        {\begin{bmatrix}M^{\rT}\Pi & 0\end{bmatrix}}\rd \cW\\
\label{dQCQF}
    & =
        \big(K   -\Im ((2L+\Pi^{\rT}\Phi)^{\rT}\Omega \Pi^{\rT} M)\big)\rd t
        +
        M^{\rT} \Pi\rd W.
\end{align}
Here, $K:= \bj_t(K_0)$ and $M:= \bj_t(M_0)$ are the evolved versions of the initial perturbations from (\ref{KM0CQF}) under the unperturbed flow $\bj_t(\zeta):= \bU(t)^{\dagger}(\cI_{\cH} \ox \zeta \ox \cI_{\cF\ox \fF})\bU(t)$ of the plant-observer-field system (which is applied here to observer operators $\zeta$ on $\fH$). In (\ref{dQCQF}),  use is also made of the relation $\Im(L^{\rT}\Omega \Pi^{\rT}M) 
= -M^{\rT}\Pi JL$ which follows from (\ref{Omega}), the commutativity $[L,M^{\rT}] = 0$ and the antisymmetry of $J$. Since the criterion process $Z$ in (\ref{CQF}) does not depend explicitly on the energy operators, the corresponding formal Gateaux derivative $\cZ'$ of the cost functional can be computed by using  \cite[Theorem 2, Section VII]{V_2015}  as
\begin{equation}
\label{cZ'}
    \cZ'
    :=
    \lim_{t\to+\infty}
        \bE \phi(t),
        \qquad
        \phi:= i[Q,Z].
\end{equation}
Here, $\phi$ is the derivative process \cite{V_2015} associated with $Z$. Its expectation satisfies the integro-differential equation
\begin{equation}
\label{bEphidot}
  (\bE \phi)^{^\centerdot} = i\bE[Q,\bG(Z)] + \bE \chi(Z),
\end{equation}
where $\bG$ is the unperturbed plant-observer GKSL generator given by (\ref{bHbL}), (\ref{bG}), (\ref{bD}), and $\chi$ is an auxiliary  linear superoperator acting on plant-observer system operators $\sigma$ as
\begin{align}
\nonumber
    \chi(\sigma)
     :=&
        i[K   -\Im ((2L+\Pi^{\rT}\Phi)^{\rT}\Omega\Pi^{\rT} M),\, \sigma]\\
\label{chiZ}
        & -2\Re ([\sigma,\, (L+\Pi^{\rT}\Phi)^{\rT}]\Omega \Pi^{\rT}M).
\end{align}
Note that $\chi(\sigma)$ depends linearly on the perturbations $K$ and $M$.
Now, in addition to the plant being an OQHO, suppose  the unperturbed observer is also an OQHO with energy matrix $r\in \mS_{\nu}$  and coupling matrices $N_1\in \mR^{p \x \nu}$ and $N_2\in \mR^{\mu \x \nu}$. The corresponding observer Hamiltonian $\Gamma$ and the coupling operators in (\ref{PhiPsi}) are
\begin{equation}
\label{Gamma_Phi_Psi}
    \Gamma = \frac{1}{2} \xi^{\rT} r\xi,
    \qquad
    \Phi = N_1 \xi,
    \qquad
    \Psi = N_2 \xi.
\end{equation}
In this case, in view of (\ref{cW}) and (\ref{Delta}), the QSDE (\ref{dxi}) becomes linear:
\begin{equation}
\label{dxilin}
    \rd \xi = a \xi \rd t + b_1 \rd Y + b_2 \rd \omega,
\end{equation}
where the matrices $a \in \mR^{\nu\x \nu}$, $b_1 \in \mR^{\nu \x p}$ and $b_2\in \mR^{\nu \x \mu}$ are computed as
\begin{align}
\label{a}
    a& := 2\vartheta (r + N_1^{\rT} \Pi J \Pi^{\rT} N_1 + N_2^{\rT} \Im \mho  N_2 ),\\
\label{bb}
    b_1 & := 2\vartheta N_1^{\rT} \Pi,
    \qquad
    b_2 := 2\vartheta N_2^{\rT}.
\end{align}
Therefore, the plant-observer system is governed by a linear QSDE
\begin{equation}
\label{dcX}
    \rd \cX
    =
    \cA
    \cX \rd t
    +
    \cB
    \rd \cW,
    \quad
    \cA:=
    {\begin{bmatrix}
        A & 0\\
        b_1 C & a
    \end{bmatrix}},
    \quad
    \cB:=
    {\begin{bmatrix}
        B & 0\\
        b_1 & b_2
    \end{bmatrix}}.
\end{equation}
For what follows, both matrices $A$ and $a$ are assumed to be Hurwitz, and hence, so is $\cA$. This implies that the plant-observer system is ergodic and has a unique invariant state which is Gaussian \cite{KRP_2010} with zero mean and quantum covariance matrix $\cP+ i\bTheta \in \mH_{n+\nu}^+$. Here, $\cP \in \mS_{n+\nu}^+$ is the controllability Gramian \cite{AM_1989,KS_1972} of the pair $(\cA, \cB)$ which is  a unique solution of the ALE
\begin{equation}
\label{PALE}
    \cA \cP  + \cP  \cA^{\rT} + \cB \cB^{\rT} = 0.
\end{equation}

\begin{lem}
\label{lem:dash}
Suppose the plant and the unperturbed observer are OQHOs described by (\ref{RN}), (\ref{dX})--(\ref{C}) and (\ref{Gamma_Phi_Psi})--(\ref{bb}),
with Hurwitz matrices $A$ and $a$. Also, let the observer Hamiltonian and the observer-plant coupling operators are perturbed according to (\ref{KM0CQF}). Then the corresponding formal Gateaux derivative of the cost functional in  (\ref{CQF}) can be  computed as
\begin{equation}
\label{cZ'CQF}
    \cZ'
    =
        \bE_* \chi(\sigma),
        \qquad
        \sigma := \cX^{\rT} \cQ \cX.
\end{equation}
Here, $\cQ \in \mS_{n+\nu}^+$ is the observability Gramian of the pair $(\cA, \cC)$ satisfying the ALE
\begin{equation}
\label{QALE}
    \cA^{\rT} \cQ  + \cQ  \cA + \cC^{\rT} \cC = 0,
\end{equation}
where the matrices $\cC$ and $\cA$ are given by (\ref{ZCQF}) and (\ref{dcX}). Also,
$\bE_*(\cdot)$ denotes the quantum expectation over the invariant Gaussian state of the plant-observer system (with $\bE_* \cX = 0$ and $\cP :=\Re \bE_* (\cX\cX^{\rT})$ found from (\ref{PALE})). Furthermore, the superoperator  $\chi$ from (\ref{chiZ}) acts on the operator $\sigma$ in (\ref{cZ'CQF}) as
\begin{align}
\nonumber
    \chi(\sigma)
     =&
        i\Big[K   -\Im
            \Big(
                \cX^{\rT}
                {\begin{bmatrix}
                    2N^{\rT}\\
                    N_1^{\rT}\Pi
                \end{bmatrix}}
                \Omega\Pi^{\rT} M
            \Big),\,
            \sigma\Big]\\
\label{chisigma}
        & +8\Im \Big(\cX^{\rT} \cQ \bit{\Theta}{\begin{bmatrix}
                    N^{\rT}\\
                    N_1^{\rT}\Pi
                \end{bmatrix}}
                \Omega \Pi^{\rT}M\Big).
\end{align}
\end{lem}
\begin{proof}
Since $Z$ in (\ref{ZCQF}) is a quadratic form of $\cX$, then (\ref{cZ'}) and (\ref{bEphidot}) imply that
\begin{equation}
\label{ZUps}
    \cZ' = \bra \cC^{\rT}\cC, \Ups\ket,
    \qquad
    \Ups := i\lim_{t\to +\infty} \bE[Q,\Xi],
\end{equation}
where $\Ups \in \mS_{n+\nu}$ is the limit mean value of the derivative process $i[Q,\Xi]$ associated with $\Xi:=\Re(\cX\cX^{\rT}) = \cX\cX^{\rT} - i\bTheta$. The GKSL generator $\bG$ of the unperturbed linear plant-observer system acts on $\Xi$ as
$
    \bG(\Xi) = \cA \Xi + \Xi \cA^{\rT} + \cB\cB^{\rT}
$,
and hence, similarly to \cite[Example 3, Section VII]{V_2015},  the matrix $\Ups$ satisfies the ALE
\begin{equation}
\label{UpsALE}
    \cA \Ups + \Ups \cA^{\rT} + \bE_* \chi(\Xi) = 0,
\end{equation}
where the superoperator $\chi$ in (\ref{chiZ}) is applied to $\Xi$ entrywise.
By combining the integral representations of the solutions of the ALEs (\ref{QALE}), (\ref{UpsALE})  and using duality, it follows from (\ref{ZUps}) that
\begin{align*}
    \cZ'
    & =
    \Big\bra
    \cC^{\rT}\cC,\,  \int_0^{+\infty}\re^{t\cA} \bE_* \chi(\Xi) \re^{t\cA^{\rT}}\rd t
    \Big\ket\\
    & =
    \Big\bra
        \int_0^{+\infty}\re^{t\cA^{\rT}} \cC^{\rT}\cC\re^{t\cA}\rd t, \,
        \bE_* \chi(\Xi)
    \Big\ket\\
    & =
    \bra
        \cQ ,\,  \bE_* \chi(\Xi)
    \ket
    =
        \bE_* \chi(\cX^{\rT} \cQ  \cX),
\end{align*}
which establishes (\ref{cZ'CQF}) since $\chi(\Xi) = \chi(\cX\cX^{\rT})$. Indeed, $\sigma$ enters $\chi(\sigma)$ in (\ref{chiZ}) only through the commutators $[\sigma, \cdot]$ and hence, $\chi(\sigma + \theta) = \chi(\sigma)$ for any $\theta \in \mC$. The representation (\ref{chisigma}) is obtained by substituting $L$ from (\ref{RN}), $\Phi$, $\Psi$  from (\ref{Gamma_Phi_Psi}) and $\sigma$ from (\ref{cZ'CQF}) into (\ref{chiZ}), and using the relation $\frac{i}{2}[\sigma, \cX] = 2\bTheta \cQ \cX$.
In particular, the second line of (\ref{chisigma}) is established by
\begin{align*}
-2\Re ([\sigma,\, &(L+\Pi^{\rT}\Phi)^{\rT}]\Omega \Pi^{\rT}M)\\
& =
-2\Re \Big([\sigma, \cX^{\rT}]{\begin{bmatrix}N^{\rT}\\ N_1^{\rT} \Pi \end{bmatrix}}\Omega \Pi^{\rT}M\Big)\\
& =
-2\Re \Big((-4i\bTheta \cQ \cX)^{\rT}{\begin{bmatrix}N^{\rT}\\ N_1^{\rT} \Pi \end{bmatrix}}\Omega \Pi^{\rT}M\Big)\\
& =
8\Im \Big(\cX^{\rT} \cQ \bTheta {\begin{bmatrix}N^{\rT}\\ N_1^{\rT} \Pi \end{bmatrix}}\Omega \Pi^{\rT}M\Big).
\end{align*}
\end{proof}

Although the unperturbed observer in Lemma~\ref{lem:dash} is an OQHO, the perturbations (\ref{KM0CQF})   are not assumed to be linear-quadratic. Therefore, the lemma provides a perturbative tool to develop conditions for such an observer to be a stationary point of the CQF problem (\ref{CQF}) in a wider  class of observers.

\section{OPTIMALITY AMONG LINEAR OBSERVERS}\label{sec:linopt}

We will now apply Lemma~\ref{lem:dash} to
the first-order necessary conditions of optimality among linear observers. 
To this end, associated with the Gramians $\cP$ and $\cQ$ from (\ref{PALE}) and (\ref{QALE}) is the Hankelian
\begin{equation}
\label{cE}
    \cE
    :=
    {\begin{bmatrix} \cE_{11} & \cE_{12}\\ \cE_{21} & \cE_{22}\end{bmatrix}}
    :=
    \cQ\cP
\end{equation}
which, together with $\cP$ and $\cQ$, is split into appropriately dimensioned blocks $(\cdot)_{jk}$ according to the partitioning of $\cX$ in (\ref{cX}), 
with $(\cdot)_{j \bullet}$ the $j$th block-row  and  $(\cdot)_{\bullet k}$  the $k$th block-column.

\begin{theorem}
\label{th:optlin}
Under the assumptions of Lemma~\ref{lem:dash}, the linear observer, described in the lemma, is a stationary point of the CQF problem (\ref{CQF}) in the class of linear observers if and only if
\begin{align}
\label{stat1}
    \vartheta \cE_{22} & \in \mA_{\nu},    \\
\label{stat2}
\Pi(C\cE_{21}^{\rT} + B^{\rT}\cQ_{12}+b_1^{\rT} \cQ_{22}) \vartheta & = \Pi J b_1^{\rT} \cE_{22}.
\end{align}
\end{theorem}
\begin{proof}
Let the matrices $r$ and $N_1$ in (\ref{Gamma_Phi_Psi}) be smooth functions of the parameter $\eps$, and hence, the
corresponding perturbations of the linear observer  in (\ref{KM0CQF})  take the form
\begin{equation}
\label{KMlin}
    K
    =
    \frac{1}{2} \xi^{\rT} r'\xi,
    \qquad
    M = N_1' \xi,
\end{equation}
where the matrices $r'\in \mS_{\nu}$ and $N_1' \in \mR^{p\x \nu}$ can be arbitrary. Substitution of (\ref{KMlin}) into (\ref{chisigma}) leads to
\begin{align}
\nonumber
    \chi(\sigma)
     =&
        i\Big[
        \frac{1}{2} \xi^{\rT} r'\xi   -\Im
            \Big(
                \cX^{\rT}
                {\begin{bmatrix}
                    2N^{\rT}\\
                    N_1^{\rT}\Pi
                \end{bmatrix}}
                \Omega\Pi^{\rT} N_1' \xi
            \Big),\,
            \sigma\Big]\\
\label{chiXX}
        & +8\Im \Big(\cX^{\rT} \cQ \bit{\Theta}{\begin{bmatrix}
                    N^{\rT}\\
                    N_1^{\rT}\Pi
                \end{bmatrix}}
                \Omega \Pi^{\rT}N_1' \xi\Big),
\end{align}
where $\sigma$ is the quantum variable from (\ref{cZ'CQF}). The identity
$\frac{i}{2}[\cX^{\rT} R_1 \cX, \cX^{\rT} R_2 \cX] = 2\cX^{\rT} (R_2\bTheta R_1 - R_1\bTheta R_2)\cX$,  which holds for any $R_1, R_2 \in \mS_{n+\nu}$ due to the CCRs (\ref{cX}), implies that
\begin{align*}
    \frac{i}{2} [\xi^{\rT} r'\xi, \sigma]
    & =
    2
    \cX^{\rT}
    \Big(
        \cQ
        \bTheta
        {\begin{bmatrix}
            0 & 0\\
            0 & r'
        \end{bmatrix}}
        -
        {\begin{bmatrix}
            0 & 0\\
            0 & r'
        \end{bmatrix}}
        \bTheta
        \cQ
    \Big)
    \cX\\
    & =
    2
    \cX^{\rT}
    \Big(
        {\begin{bmatrix}
            0 & \cQ_{\bullet 2}\vartheta r'
        \end{bmatrix}}
        -
        {\begin{bmatrix}
            0\\
            r'\vartheta \cQ_{2 \bullet}
        \end{bmatrix}}
    \Big)
    \cX.
\end{align*}
The averaging of the latter representation over the invariant quantum state leads to
\begin{align*}
\nonumber
    \frac{i}{2}
    \bE_*[\xi^{\rT} r'\xi, \sigma]
    & =
    2
    \Big\bra
        {\begin{bmatrix}
            0 & \cQ_{\bullet 2}\vartheta r'
        \end{bmatrix}}
        -
        {\begin{bmatrix}
            0\\
            r'\vartheta \cQ_{2 \bullet}
        \end{bmatrix}},\,
        \cP
    \Big\ket\\
    & =
    -4
    \Bra
        \vartheta \cQ_{2\bullet} \cP_{\bullet 2},\,
        r'
    \Ket
    =
    -4
    \Bra
        \bS(\vartheta \cE_{22}),\,
        r'
    \Ket,
\end{align*}
where the equality $\cQ_{2\bullet} \cP_{\bullet 2} = \cE_{22}$ follows from (\ref{cE}), and the symmetry of $r'$ is used.
Therefore, in view of (\ref{cZ'CQF}),  the corresponding formal Frechet derivative $\d_r \cZ =     -4\bS(\vartheta \cE_{22})$ of the cost functional vanishes if and only if the matrix $\cE_{22}$ satisfies  (\ref{stat1}). By a similar reasoning, in view of (\ref{chiXX}), the Gateaux derivative of $\cZ$ along $M$  in (\ref{KMlin}) takes the form
\begin{align}
\nonumber
         8\Im &\bE_*\Big(\cX^{\rT} \cQ \bit{\Theta}{\begin{bmatrix}
                    N^{\rT}\\
                    N_1^{\rT}\Pi
                \end{bmatrix}}
                \Omega \Pi^{\rT}N_1' \xi\Big)\\
\nonumber
         & -\Re \bE_*\Big[
                \cX^{\rT}
                {\begin{bmatrix}
                    2N^{\rT}\\
                    N_1^{\rT}\Pi
                \end{bmatrix}}
                \Omega\Pi^{\rT} N_1' \xi
            ,\,
            \sigma\Big]\\
\nonumber
          =& 8\Im
         \Bra
\cQ \bit{\Theta}{\begin{bmatrix}
                    N^{\rT}\\
                    N_1^{\rT}\Pi
                \end{bmatrix}}
                \overline{\Omega} \Pi^{\rT}N_1', \,
                \cP_{\bullet 2}
                +
                {\begin{bmatrix}
                    0\\
                    i\vartheta
                \end{bmatrix}}
         \Ket
         \\
\nonumber
         & +
         4
         \Re
         \Bra
            {\begin{bmatrix}
                0 &
            {\begin{bmatrix}
                    2N^{\rT}\\
                    N_1^{\rT}\Pi
                \end{bmatrix}}
                \overline{\Omega}\Pi^{\rT} N_1'
            \end{bmatrix}},\,
            i(\cE^{\rT} \bTheta-\bTheta \cE)
         \Ket
          \\
\nonumber
    =&
    8
    \Big\bra
        \Pi J {\begin{bmatrix}
                    N &
                    \Pi^{\rT}N_1
                \end{bmatrix}}
                \bTheta \cE_{\bullet 2}
                -
                \Pi {\begin{bmatrix}
                    N &
                    \Pi^{\rT}N_1
                \end{bmatrix}}
                \bTheta \cQ_{\bullet 2}\vartheta\\
\nonumber
    & +
        \Pi JN(\cE_{21}^{\rT} \vartheta - \Theta \cE_{12}) - \Pi J \Pi^{\rT} N_1 \bS(\vartheta \cE_{22}),\,
                N_1'
    \Big\ket\\
\label{chain1}
    = &
    4
    \Bra
        \Pi( (C\cE_{21}^{\rT} + B^{\rT}\cQ_{12}+b_1^{\rT} \cQ_{22}) \vartheta -J b_1^{\rT} \cE_{22}),\,
        N_1'
    \Ket,
\end{align}
where use is also made of (\ref{AB}), (\ref{C}), (\ref{bb}) and (\ref{cE}). Here, $\bS(\vartheta \cE_{22}) = 0$ under the condition (\ref{stat1}), in which case, the formal Frechet derivative
$
   \d _{N_1} \cZ = 4 \Pi( (C\cE_{21}^{\rT} + B^{\rT}\cQ_{12}+b_1^{\rT} \cQ_{22}) \vartheta -J b_1^{\rT} \cE_{22})
$ vanishes if and only if (\ref{stat2}) holds. Therefore, (\ref{stat1}) and (\ref{stat2}) are indeed equivalent to the stationarity in the class of linear observers.  \end{proof}


\section{WEYL VARIATIONS OF OBSERVERS
}\label{sec:Weyl}

Now, consider the following bounded perturbations of the observer Hamiltonian and the observer-plant coupling operators in (\ref{KM0CQF}):
\begin{align}
\label{Kex}
    K & := \Re (\alpha \sW_u) = \frac{1}{2} (\alpha \sW_u + \overline{\alpha}\sW_{-u})=|\alpha|\cos(u^{\rT}\xi + \gamma),\\
\label{Mex}
    M & := \Re (\beta\sW_u) = \frac{1}{2} (\beta \sW_u + \overline{\beta}\sW_{-u})
      = (|\beta_k| \cos(u^{\rT}X + \theta_k))_{1\<k\< p},
\end{align}
with parameters $\alpha\in \mC$,  $\beta:= (\beta_k)_{1\< k\< p} \in \mC^p$, $\gamma:= \arg \alpha$, $\theta_k := \arg \beta_k$ and $u\in \mR^{\nu}$. Here,
$
    \sW_u:= \re^{iu^{\rT} \xi} = \sW_{-u}^{\dagger}
$
is the unitary Weyl operator 
associated with the observer variables. We will refer to the perturbations (\ref{Kex}) and (\ref{Mex}) as the Weyl variations (of the observer energy operators). Similar trigonometric polynomials of quantum variables have recently been used in \cite{SVP_2014} to model uncertainties in system Hamiltonians. However, in what follows, the Weyl variations play a different role which resembles that of the needle variations in the proof of the Pontryagin minimum principle \cite{PBGM_1962}.

\begin{theorem}
\label{th:weyl}
Under the assumptions of Lemma~\ref{lem:dash}, the linear observer satisfies the first-order necessary conditions of optimality in the CQF problem (\ref{CQF}) among linear observers if and only if it is a stationary point with respect to the Weyl variations in (\ref{Kex}) and (\ref{Mex}) for any $\alpha\in \mC$, $\beta\in \mC^p$, $u\in \mR^{\nu}$.
\end{theorem}
\begin{proof}
The stationarity of the cost functional $\cZ$ with respect to arbitrary Weyl variations of the observer implies the stationarity with respect to the linear-quadratic perturbations (\ref{KMlin}) because  $\d_u\sW_u\big|_{u=0} = i\xi$ and $\d_u^2\sW_u\big|_{u=0} = -\Re(\xi\xi^{\rT})$. In fact, arbitrary polynomial perturbations can be reproduced by the Weyl variations which form a larger class of perturbations. Therefore, it remains to prove that (\ref{stat1}) and (\ref{stat2}) imply the stationarity with respect to arbitrary Weyl variations.
By Lemma~\ref{lem:dash}, the linear response of the cost functional $\cZ$ to the Weyl variation $K$ of the observer Hamiltonian in (\ref{Kex}) takes the form
\begin{equation}
\label{Kvar}
    i
    \bE_*
    [\Re (\alpha \sW_u), \sigma]
    =
    -
    \Im (
        \alpha
        \bE_*
        [\sW_u, \sigma]
        ).
\end{equation}
Let $\sS_u$ denote a linear superoperator, which is parameterized by $u \in \mR^{\nu}$ and carries out a unitary similarity transformation of quantum variables as
$
    \sS_u(\zeta) := \sW_u \zeta\sW_{-u} = \re^{i \ad_{u^{\rT} \xi}}(\zeta) = \sum_{k=0}^{+\infty} \frac{i^k}{k!}\ad_{u^{\rT} \xi}^k(\zeta)
$, where use is made of  a well-known identity for operator exponentials  \cite{M_1998,W_1967}. Its application to $\cX$ from (\ref{cX}) leads to
$
    \sS_u(\cX)
    =
    \cX + {\scriptsize\begin{bmatrix}
    0\\
    2\vartheta u
    \end{bmatrix}}
$
since $[X, \xi^{\rT}] = 0$ (whence $\sS_u(X) = X$) and $[u^{\rT}\xi,\xi] = -[\xi,u^{\rT}\xi] =-[\xi,\xi^{\rT}]u = -2i\vartheta u$  (whereby $\sS_u(\xi) = \xi + 2\vartheta u$). Therefore,  in view of (\ref{cZ'CQF}),  $\sS_u(\sigma) = \sS_u(\cX)^{\rT} \cQ \sS_u(\cX) = \sigma-4u^{\rT}\vartheta (\cQ_{2\bullet}\cX + \cQ_{22} \vartheta u)$, and hence,
\begin{align}
\nonumber
    [\sW_u, \sigma]
    & = (\sS_u(\sigma) - \sigma)\sW_u\\
\nonumber
     & = -4u^{\rT}\vartheta (\cQ_{21}X + \cQ_{22}(\xi+ \vartheta u))
    \sW_u\\
\label{Wsig}
    & =
    4iu^{\rT}\vartheta
    (\cQ_{21}\d_w + \cQ_{22}\d_u)
    \re^{i(w^{\rT} X + u^{\rT}\xi)}\big|_{w=0}.
\end{align}
Here, $\re^{iw^{\rT} X}\sW_u = \re^{i(w^{\rT} X + u^{\rT}\xi)}$ in view of the commutativity $[\re^{iw^{\rT} X}, \sW_u]=0$ between the Weyl operators for all $w\in \mR^n$, $u\in \mR^{\nu}$, and hence,
\begin{equation}
\label{XW}
    X\sW_u = -i\d_w \re^{iw^{\rT} X}\big|_{w=0} \sW_u
    =
    -i\d_w \re^{i(w^{\rT} X + u^{\rT}\xi)}\big|_{w=0}.
\end{equation}
Also, the relation $(\xi + \vartheta u) \sW_u = -i\d_u\sW_u$  follows from
\begin{align}
\nonumber
    \xi\sW_u
    & =
    -i \d_v \sW_v \big|_{v=0} \sW_u\\
\label{xiW}
    & =
    -i \d_v(    \re^{i u^{\rT}\vartheta v}  \sW_{u+v})\big|_{v=0}
     = -\vartheta u \sW_u - i\d_u\sW_u
\end{align}
due to the Weyl CCRs
$
\sW_v \sW_u
=
\re^{i u^{\rT}\vartheta v}  \sW_{u+v}
$ which hold for all $u,v\in \mR^{\nu}$. By averaging (\ref{Wsig}) over the invariant Gaussian quantum state, whose
quasi-characteristic function \cite{CH_1971} is
$    \bE_*\re^{i\lambda^{\rT} \cX}
    =
    \re^{-\frac{1}{2}\|\lambda\|_{\cP}^2}
$ for any $\lambda:= {\scriptsize\begin{bmatrix}w\\ u\end{bmatrix}}\in \mR^{n+\nu}$, it follows that
\begin{align}
\nonumber
    \bE_*
    [\sW_u, \sigma]
    & =
    4iu^{\rT}\vartheta
    \cQ_{2\bullet}
    \d_{\lambda}
    \re^{-\frac{1}{2}\|\lambda\|_{\cP}^2}
    \big|_{w=0}\\
\nonumber
    & =
    -4iu^{\rT}\vartheta
    \cQ_{2\bullet}
    \cP
    {\begin{bmatrix}0\\ u\end{bmatrix}}
    \re^{-\frac{1}{2}\|u\|_{\cP_{22}}^2}\\
\label{Kvargauss}
    & =
    -4iu^{\rT}
    \bS(\vartheta \cE_{22})
    u
    \re^{-\frac{1}{2}\|u\|_{\cP_{22}}^2},
\end{align}
where use is made of the Hankelian from (\ref{cE}).  By substituting (\ref{Kvargauss}) into (\ref{Kvar}), the Gateaux derivative of $\cZ$ along $K$ in (\ref{Kex}) takes the form
$
    i
    \bE_*
    [K, \sigma]
    =
    4u^{\rT}
    \bS(\vartheta \cE_{22})
    u
    \re^{-\frac{1}{2}\|u\|_{\cP_{22}}^2}
    \Re \alpha
$.
Therefore, the fulfillment of (\ref{stat1}) makes
this derivative vanish for any $\alpha\in \mC$ and $u \in \mR^{\nu}$. Similar arguments (see also (\ref{chain1})) lead to the following Gateaux derivative of $\cZ$ along $M$ in (\ref{Mex}):
\begin{align}
\nonumber
         8&\Im \bE_*\Big(\cX^{\rT} \cQ \bit{\Theta}{\begin{bmatrix}
                    N^{\rT}\\
                    N_1^{\rT}\Pi
                \end{bmatrix}}
                \Omega \Pi^{\rT}\Re (\beta\sW_u)\Big)\\
\nonumber
         & -\Re \bE_*\Big[
                \cX^{\rT}
                {\begin{bmatrix}
                    2N^{\rT}\\
                    N_1^{\rT}\Pi
                \end{bmatrix}}
                \Omega\Pi^{\rT} \Re (\beta\sW_u)
            ,\,
            \sigma\Big]\\
\nonumber
          =& 8\Im
         \Bra
\cQ \bit{\Theta}{\begin{bmatrix}
                    N^{\rT}\\
                    N_1^{\rT}\Pi
                \end{bmatrix}}
                \overline{\Omega} \Pi^{\rT}, \,
                \bE_*(\cX \Re (\beta^{\rT}\sW_u))
         \Ket
         \\
\nonumber
         & +
         \Re
         \Bra
            {\begin{bmatrix}
                    2N^{\rT}\\
                    N_1^{\rT}\Pi
                \end{bmatrix}}
                \overline{\Omega}\Pi^{\rT}, \,
            \bE_*[\sigma, \cX \Re (\beta^{\rT}\sW_u)]
         \Ket
          \\
\nonumber
          =& -8
          \re^{-\frac{1}{2}\|u\|_{\cP_{22}}^2}\Im
         \Bra
\cQ \bit{\Theta}{\begin{bmatrix}
                    N^{\rT}\\
                    N_1^{\rT}\Pi
                \end{bmatrix}}
                \overline{\Omega} \Pi^{\rT},
                \Big(\cP_{\bullet 2}
                +
                {\begin{bmatrix}
                    0\\
                    i\vartheta
                \end{bmatrix}}\Big)u\Im \beta^{\rT}
         \Ket
         \\
\nonumber
         & +
         4 \re^{-\frac{1}{2}\|u\|_{\cP_{22}}^2}
         \Re
         \Bra
            {\begin{bmatrix}
                    2N^{\rT}\\
                    N_1^{\rT}\Pi
                \end{bmatrix}}
                \overline{\Omega}\Pi^{\rT}, \,
            i(\Theta \cE_{12}-\cE_{21}^{\rT} \vartheta)u\Im \beta^{\rT}
         \Ket
          \\
\label{chain2}
    = &
    4
          \Im \beta^{\rT}
        \Pi(J b_1^{\rT} \cE_{22}- (C\cE_{21}^{\rT} + B^{\rT}\cQ_{12}+b_1^{\rT} \cQ_{22}) \vartheta)u
        \re^{-\frac{1}{2}\|u\|_{\cP_{22}}^2}.
\end{align}
Here, $\bE_*[\sigma, \sW_u]=0$ follows from (\ref{Kvargauss}) under the condition (\ref{stat1}).
Also, in view of (\ref{Wsig})--(\ref{xiW}), use is made of the relations
\begin{align*}
    \bE_*(\cX \sW_u)
    & =
    -i\d_{\lambda} \bE_* \re^{i\lambda^{\rT}\cX}\big|_{w=0}
    -{\begin{bmatrix}0\\ \vartheta u\end{bmatrix}}\bE_* \sW_u\\
    & =
    \re^{-\frac{1}{2}\|u\|_{\cP_{22}}^2}
    \Big(i\cP_{\bullet 2}  -{\begin{bmatrix}0\\ \vartheta \end{bmatrix}}\Big)u,\\
    \bE_*[\sigma, \cX \sW_u]
    & =
    -i\d_{\lambda} \bE_* [\sigma, \re^{i\lambda^{\rT}\cX}]\big|_{w=0}
    -{\begin{bmatrix}0\\ \vartheta u\end{bmatrix}}\bE_* [\sigma, \sW_u]\\
    & =
    -4\d_{\lambda} (\lambda^{\rT}\bTheta \cQ \d_{\lambda}\bE_* \re^{i\lambda^{\rT}\cX})\big|_{w=0}\\
    & =
    4\d_{\lambda} \big(\lambda^{\rT}\bS(\bTheta \cE)\lambda \re^{-\frac{1}{2}\|\lambda\|_{\cP}^2}\big)\big|_{w=0}\\
    & =
    4 \re^{-\frac{1}{2}\|u\|_{\cP_{22}}^2}
    (\Theta \cE_{12}-\cE_{21}^{\rT} \vartheta)u.
\end{align*}
The fulfillment of (\ref{stat1}) and (\ref{stat2}) makes the Gateaux derivative (\ref{chain2}) of $\cZ$ along the Weyl variation $M$ in (\ref{Mex}) also vanish for any $\beta\in \mC^p$ and  $u \in \mR^{\nu}$. Therefore, (\ref{stat1}) and (\ref{stat2}) indeed imply the stationarity with respect to the arbitrary Weyl variations (\ref{Kex}) and (\ref{Mex}).
\end{proof}

Note that the proof of Theorem~\ref{th:weyl} (which employs the Gaussian averaging of quasi-polynomials of system variables, such as $\cX \re^{i\lambda^{\rT}\cX}$) is closely related to the integro-differential identities  \cite{V_2014b} for expectations of Weyl quantization integrals over Gaussian states.

\section{CONCLUSION}\label{sec:conc}

We have applied the transverse Hamiltonian variational method \cite{V_2015} to the mean square optimal CQF problem for a linear observer cascaded with a linear quantum plant. It has been shown that if such an observer is a stationary point of the problem among linear observers, then it also satisfies the first-order  necessary conditions of optimality with respect to a wider class of Weyl variations of the energy operators of the observer. In this sense,  linear observers are locally sufficient  for linear quantum plants as far as the mean square performance criteria are concerned.
Similar ideas can be developed for the CQLQG control problem \cite{NJP_2009} and an optimal control theory for classical port-Hamiltonian systems \cite{V_2006} and their stochastic versions.

\end{document}